\documentclass{fundam}

\publyear{26}
\papernumber{3}
\volume{196}
\issue{2}
\theDOI{10.46298/fi.16091}

\versionForARXIV

\usepackage{url} 
\usepackage{amsmath,amssymb,enumerate}

\newcommand{\cP}{\mathcal P}
\newcommand{\cS}{\mathcal S}
\newcommand{\cF}{\mathcal F}
\newcommand{\cR}{\mathcal R}
\newcommand{\cA}{\mathcal A}
\newcommand{\cH}{\mathcal H}
\newcommand{\cK}{\mathcal K}
\newcommand{\cL}{\mathcal L}

\newcommand{\Id}{\mathrm{Id}}
\newcommand{\bN}{\mathbb N}
\newcommand{\bM}{\mathbb M}
\newcommand{\ConfGraph}{\mathbb G}
\newcommand{\Conf}{\mathrm{Conf}}
\newcommand{\state}{\mathrm{state}}
\newcommand{\ol}[1]{\overline{#1}}
\newcommand{\pair}[2]{
  \begin{pmatrix}
    #1\\#2
  \end{pmatrix}
}
\DeclareMathOperator{\twins}{\mathrm{twns}}

\DeclareMathOperator{\alphabet}{\mathrm{alph}}
\DeclareMathOperator{\reach}{\mathrm{Reach}}
\DeclareMathOperator{\step}{\mathrm{Step}}
\DeclareMathOperator{\fnf}{\mathrm{fnf}}
\DeclareMathOperator{\eFNF}{\mathrm{eFNF}}

\begin{document}

\title{The theory of reachability of trace-pushdown
  systems\thanks{\ This is an extended version of the author's
    conference publication \cite{Kus25} and contains full proofs.}}


\author{Dietrich Kuske\\
Technische Universi\"at Ilmenau\\
dietrich.kuske{@}tu-ilmenau.de} 

\maketitle

\runninghead{D. Kuske}{The theory of reachability of trace-pushdown
  systems}

\begin{abstract}
  We consider pushdown systems that store, instead of a single word, a
  Mazurkiewicz trace on its stack. These systems are special cases of
  valence automata over graph monoids and subsume multi-stack
  systems. We identify a class of such systems whose configuration
  graph with reachability has a decidable first-order theory.

  This result complements results by D'Osualdo, Meyer, and Zetzsche
  (namely the decidability for arbitrary pushdown systems under a
  severe restriction on the dependence alphabet).
\end{abstract}


\section{Introduction}

Pushdown systems are a well-studied type of infinite state system with
many variants and generalizations.  Valence automata
\cite{IbaSK76,Gil94,ItoMVM01,MitS01,ElsO04,EldKO08,Kam09,REnK09,Ren10,Zet15,Zet21,OsuMZ16}
can be used to study many of these variants in a uniform setting: in a
pushdown system, the possible pushdown contents come from a free
monoid and can be accessed at the suffix, only. Differently, in a
valence automaton, the free monoid is replaced by an arbitrary
monoid. As an example, graph monoids allow us to model pushdown
systems, Petri nets, multi-stack automata, counter automata, and many
more.

For a pushdown system, the configuration graph with reachability can
be interpreted in the full tree. Hence Rabin's tree theorem ensures
that monadic second order properties of this graph are decidable. This
decidability gets lost very soon once one considers valence automata
over graph monoids. Such a valence automaton consists of a finite
control and the graph monoid as generalization of the pushdown.

Therefore, D'Osualdo, Meyer and Zetzsche \cite{OsuMZ16} considered
first-order logic. They were able to characterize those graph
monoids  that ensure the decidability, no matter what local control
the monoid is attached to.

Differently here, we allow all loop-free graph monoids, but aim at
properties of the finite control that ensure the decidability
of the first-order theory of the configuration graph with
reachability. These loop-free graph monoids are precisely
Mazurkiewicz' trace monoids.

In \cite{KoeK23,KoeK25a,KoeK26}, we considered a related problem:
there, we were concerned with the question what properties of the
finite control ensure the decidability of the reachability
relation. It turned out that this is the case for what we called
cooperating multi-pushdown systems in \cite{KoeK23,KoeK25a} and, more
generally, trace-pushdown systems in \cite{KoeK26}. To prove this
decidability, we showed that, in this case, the reachability relation
(a binary relation on the trace monoid) is ``lc-rational''
\cite{Kus23}, a non-trivial restriction of rational trace
relations. In general terms, this proof proceeds as follows. One first
handles two very special cases of systems where, either, the pushdown
is shortened in every step (i.e., no transition writes anything to the
pushdown) or, alternatively, the pushdown is never shortened, but all
transitions read ``similar'' letters (in some precise meaning). The
central argument in these two cases is that the class of lc-rational
relations is closed under the componentwise concatenation with direct
products of a recognizable and a rational trace language (in general,
the concatenation of lc-rational relations need not be
lc-rational). To handle the general case, one shows that the
reachability relation can be obtained from the special cases by union
and composition. Here, one uses that the class of lc-rational
relations is closed under these two operations (in general, the
composition of rational trace relations is not necessarily rational).

A first observation in this paper is that the said first-order theory
is in general undecidable for trace-pushdown systems, the
undecidability even holds for existential formulas
(Example~\ref{E-undecidable}). To overcome this undecidability, we
restrict attention to loop-connected trace-pushdown systems. The
configuration graph of such systems is proved to be an automatic
structure \cite{Hod82,KhoN95,BluG00} implying the desired
decidability.

The central argument in this proof is therefore the automaticity of
the configuration graph of a loop-connected trace-pushdown system. To
this aim, we use the Foata normal form of a trace as its word
representation. Consequently, one has to verify that the reachability
relation between such Foata normal forms can be accepted by a
synchronous two-tape automaton, i.e., that the reachability relation
is ``fnf-automatic''. This proof follows the line of argument of the
lc-rationality for arbitrary trace-pushdown systems, but extends it at
some central point. So one first has to handle the two special cases
mentioned above. In both these cases, it is used that the class of
fnf-automatic relations is closed under concatenation with direct products
of two recognizable trace languages (Lemma~\ref{L-fnf-sr}). Then the
result for deleting systems follows immediately. For non-deleting
systems, the new requirement to be loop-connected allows to show that
the set of traces that can be written onto the pushdown is
recognizable (without the requirement, it was known to be rational,
only).  Now the general case follows as before since the class of
fnf-automatic relations is closed under composition and union.

\section{Preliminaries}

\subsection{Mazurkiewicz traces}

Zero belongs to the set $\bN$ of non-negative integers and, for
$n\in\bN$, we set $[n]=\{1,2,\dots,n\}$.

Let $\Sigma$ be some alphabet and $w\in \Sigma^*$ a word over
$\Sigma$. The \emph{alphabet of $w$}, denoted
$\alphabet(w)\subseteq \Sigma$, is the set of letters occurring in the
word $w$. The number of occurrences of the letter $a$ in the word $w$
is denoted $|w|_a$.

A \emph{dependence alphabet} is a pair $(\Sigma,D)$ where $\Sigma$ is
a finite set of \emph{letters} and $D\subseteq \Sigma\times \Sigma$ is
a reflexive and symmetric relation, the \emph{dependence
  relation}. For a letter $a\in \Sigma$, we write $D(a)$ for the set
$\{b\in \Sigma \mid (a,b)\in D\}$ of letters dependent on $a$,
$D(B)=\bigcup_{b\in B}D(b)$ is the set of letters dependent on some
letter in $B\subseteq \Sigma$; $D(w)=D(\alphabet(w))$ for any word
$w\in\Sigma^*$. A set of letters $A\subseteq\Sigma$ is
\emph{connected} if the graph $(A,D\cap A^2)$ is connected.

Two letters $a$ and $b$ are \emph{twins} if $D(a)=D(b)$. The set of
twins of a letter $a$ is denoted
$\twins(a)=\{b\in\Sigma\mid D(a)=D(b)\}$. The \emph{twin index} of
$(\Sigma,D)$ is the number of sets $\twins(a)$ with $a\in\Sigma$.

The \emph{independence relation} $I\subseteq \Sigma\times \Sigma$ is
the set of pairs $(a,b)$ of distinct letters with $(a,b)\notin D$. If
$A$ and $B$ are two sets of letters with $A\times B\subseteq I$ (i.e.,
any letter from $A$ is independent from any letter from $B$), then we
write $A\parallel B$. For words $u,v\in\Sigma^*$, $u\parallel v$
abbreviates $\alphabet(u)\parallel\alphabet(v)$.

Let ${\sim}\subseteq \Sigma^*\times \Sigma^*$ denote the least monoid
congruence with $ab\sim ba$ for all $(a,b)\in I$. In other words,
$u\sim v$ holds for two words $u,v\in \Sigma^*$ iff $u$ can be
obtained from $v$ by transposing consecutive independent letters. In
particular, $u\sim v$ implies $|u|=|v|$, $|u|_a=|v|_a$ for all
$a\in\Sigma$, and $\alphabet(u)=\alphabet(v)$. Furthermore,
$u\parallel v$ implies $uv\sim vu$ (but the converse implication does
not hold as the example $u=v=ab$ shows).

The \emph{(Mazurkiewicz) trace monoid} induced by $(\Sigma,D)$ is the
quotient of the free monoid $\Sigma^*$ wrt.\ the congruence $\sim$,
i.e., $\bM=\Sigma^*/{\sim}$. Its elements are equivalence classes of
words denoted $[w]$; by $[w]$, we mean the equivalence class
containing $w$, it is the \emph{trace induced by $w$}. Note that $\bM$
is indeed a monoid whose unit element is the trace $[\varepsilon]$;
this trace is denoted $1$.

Suppose $D=\Sigma\times \Sigma$, i.e., all letters are mutually dependent. Then
$u\sim v\iff u=v$ holds for all words $u,v\in \Sigma^*$; hence
$\bM\cong \Sigma^*$ in this case.

The other extreme is $D=\{(a,a)\mid a\in \Sigma\}$ where any two
distinct letters are independent. Then $u\sim v$ iff $|u|_a=|v|_a$
holds for any letter $a\in \Sigma$. Hence $\bM\cong(\bN^{|\Sigma|},+)$
in this case.

To consider a case in between, let $\Sigma$ be the disjoint union of
$A$ and $B$ and suppose $D=A^2\cup B^2$. Then $\bM$ is isomorphic to
the direct product of the free monoids $A^*$ and $B^*$.

\paragraph{Automata and Word Languages} An \emph{$\varepsilon$-NFA}
or \emph{nondeterministic finite automaton with
  $\varepsilon$-transitions} is a tuple $\cA=(Q,\Sigma,I,T,F)$ where
$Q$ is a finite set of \emph{states}, $\Sigma$ is an alphabet,
$I,F\subseteq Q$ are the sets of \emph{initial} and \emph{final}
states, respectively, and
\[
  T\subseteq Q\times (\Sigma\cup\{\varepsilon\}) \times Q
\]
is the set of \emph{transitions}.  The $\varepsilon$-NFA $\cA$ is an
\emph{NFA} if $T\subseteq Q\times \Sigma\times Q$.

Let $\cA=(Q,\Sigma,I,T,F)$ be an $\varepsilon$-NFA. A \emph{path} is a
sequence
\[
  (p_0,a_1,p_1)(p_1,a_2,p_2)\cdots (p_{n-1},a_n,p_n)
\]
of matching transitions (i.e., elements of $T$). Such a path is usually denoted
\[
  p_0 \xrightarrow{a_1} p_1
  \xrightarrow{a_2} p_2
  \cdots
  \xrightarrow{a_n} p_n
\]
or, if the intermediate states are of no importance,
\[
  p_0 \xrightarrow{a_1a_2\cdots a_n} p_n\,.
\]
This path is \emph{accepting} if it connects an initial state with a
final state, i.e., if $p_0\in I$ and $p_n\in F$. It accepts the word
$w=a_1a_2\cdots a_n$ (note that $a_i\in \Sigma\cup\{\varepsilon\}$
such that $|w|<n$ is possible). We denote by $L(\cA)$ the set of words
accepted by $\cA$. A language $L\subseteq \Sigma^*$ is \emph{regular}
if it is accepted by some NFA, i.e., if there is some NFA $\cA$ with
$L=L(\cA)$.

A foundational result in the theory of finite automata states that
$\varepsilon$-NFA and NFA are equally expressive. Even more,
$\varepsilon$-NFA can be transformed into equivalent NFA in polynomial
time.

A language $L\subseteq \Sigma^*$ is \emph{rational} if it can be
constructed from finite languages using the operations union,
multiplication, and Kleene star. By Kleene's theorem \cite{Kle56}, a
language is regular if, and only if, it is rational.

\paragraph{Rational and Recognizable Trace Languages}

Fix some dependence alphabet $(\Sigma,D)$.  For a word language
$L\subseteq \Sigma^*$, we denote by $[L]$ the set of traces $[u]$ induced
by words from $L$, i.e., $[L]=\{[u]\mid u\in L\}\subseteq\bM$.

Now let $\cL\subseteq\bM$. The set $\cL$ is \emph{recognizable} if the
word language $\{u\in \Sigma^*\mid [u]\in \cL\}$ is regular; it is
\emph{rational} if there exists a regular word language
$L\subseteq \Sigma^*$ with $\cL=[L]$. It follows that every
recognizable trace language is rational; the converse implication is
known to fail (consider, e.g., the trace language
$\{[ab]^n\mid n\in\bN\}$ with $(a,b)\in I$ that is rational, but not
recognizable).

Since rational trace languages are homomorphic images of rational word
languages, we obtain that $\cL\subseteq\bM$ is rational if, and only
if, it can be constructed from finite trace languages using the
operations union, multiplication, and Kleene star.

\subsection{Trace-pushdown systems}

In this paper, we consider pushdowns that hold a trace $[w]\in\bM$
which can be accessed at the suffix, only.

A \emph{pushdown system} is a tuple $\cP=(Q,\Sigma,D,\Delta)$ where
$Q$ is a finite set of \emph{states}, $(\Sigma,D)$ a dependence
alphabet, and $\Delta\subseteq Q\times\Sigma\times\Sigma^*\times Q$ a
finite set of \emph{transitions}. For $(p,a,w,q)\in\Delta$, we
regularly write $p\xrightarrow{a\mid w}q$.

The set of configurations $\Conf(\cP)$ of $\cP$ is $Q\times\bM$. For
two configurations $(p,[u]),(q,[v])\in \Conf(\cP)$, we set
$(p,[u])\vdash(q,[v])$ if there is a transition $(p,a,w,q)\in\Delta$
and a word $x\in \Sigma^*$ such that $[u]=[xa]$ and $[v]=[xw]$. Note
that $[u]=[xa]$ is equivalent to saying $[u]=[x]\cdot[a]$ and
similarly $[v]=[xw]$ is equivalent to $[v]=[x]\cdot[w]$. Hence
$(p,s)\vdash(q,t)$ for $p,q\in Q$ and $s,t\in\bM$ iff there is a
transition $(p,a,w,q)\in\Delta$ such that the trace $t$ results from
the trace $s$ by replacing the suffix $[a]$ by $[w]$. The reflexive
and transitive closure of the one-step relation $\vdash$ is denoted
$\vdash^*$. 

\begin{definition}
  Let $\cP=(Q,\Sigma,D,\Delta)$ be a pushdown system. Its
  \emph{configuration graph with reachability} is the structure
  \[
    \ConfGraph(\cP)=\bigl(Q\times\bM,\vdash,\vdash^*,
    (\state_p)_{p\in Q},(c)_{c\in Q\times\bM}\bigr)
  \]
  consisting of the set of configurations $\Conf(\cP)=Q\times\bM$, the
  one-step relation~$\vdash$, the reachability relation~$\vdash^*$,
  the unary state predicates $\state_p=\{p\}\times\bM$, and a constant
  for every configuration of $\cP$.
\end{definition}
If $D=\Sigma\times\Sigma$, then this structure can be interpreted in
the full tree. Hence, by Rabin's tree theorem \cite{Rab69}, its
monadic second order theory is decidable. This is not true in general
as the following result implies.

\begin{theorem}[cf.~D'Osualdo, Meyer, Zetzsche~\relax{\cite[Thm.~2.1]{OsuMZ16}}]
  Let $(\Sigma,D)$ be a dependence alphabet. The first-order theory of
  $\ConfGraph(\cP)$ is decidable for every pushdown system
  $\cP=(Q,\Sigma,D,\Delta)$ if, and only if, any letter from $\Sigma$
  is independent from at most one letter.
\end{theorem}

In this paper, we take another point of view: we ask for properties of
the transitions of a pushdown system $\cP$ that guarantee the
decidability of the theory of $\ConfGraph(\cP)$. The decidability of
the theory of $\ConfGraph(\cP)$ implies in particular the decidability
of the reachability relation.  \cite[Thm.~5.1]{KoeK26} presents a
class of pushdown systems whose reachability relation is
decidable. These \emph{trace-pushdown systems} are defined as follows.

\begin{definition}\label{def:cpds}
  A \emph{trace-pushdown system} (or \emph{tPDS}, for short) is a
  pushdown system $\cP=(Q,\Sigma,D,\Delta)$ such that the following
  hold:
  \begin{enumerate}
  \item[(P1)] whenever $p\xrightarrow{a\mid w}q$ is a transition from
    $\Delta$, then $D(w)\subseteq D(a)$,\label{def:tPDS1}
  \item[(P2)] for any transitions
    $p\xrightarrow{a\mid v}q\xrightarrow{b\mid w}r$ from $\Delta$ with
    $a\parallel b$, there is a state $q'\in Q$ such that
    $p\xrightarrow{b\mid w}q'\xrightarrow{a\mid v}r$ are transitions
    from $\Delta$.
  \end{enumerate}
\end{definition}

The main result about these tPDS reads as follows:

\begin{theorem}[cf.~K\"ocher and Kuske \protect{\cite[Thm.~5.1]{KoeK26}}]
  The following problem is decidable in time polynomial in the size of
  $\cP$ and exponential in the twin index of $(\Sigma,D)$:\\
  input:
  \begin{minipage}[t]{.7\linewidth}
    trace-pushdown system $\cP=(Q,\Sigma,D,\Delta)$\\
    two configurations $(p,s)$ and $(q,t)$ of $\cP$
  \end{minipage}\\
  question: Does  $(p,s)\vdash^*(q,t)$ hold?
\end{theorem}
We also constructed pushdown systems that satisfy only one of the two
conditions (P1) and (P2) and have an undecidable reachability
relation.

The following example shows that the first-order theory of
$\ConfGraph(\cP)$ is not uniformly decidable for the class of all tPDS.

\begin{example}\label{E-undecidable}
  We want to reduce Post's correspondence problem to the existential
  theory. To this aim, let $\Sigma=\{a,b,a',b',\top\}$ with
  \[
    D= \{a,b\}^2\cup\{a',b'\}^2\cup(\Sigma\times\{\top\})\cup
    (\{\top\}\times\Sigma)\,.
  \]
  Furthermore, let $f\colon\{a,b\}^*\to\{a',b'\}^*$ be the monoid
  homomorphism with $f(a)=a'$ and $f(b)=b'$.

  An instance of Post's correspondence problem is a tuple
  $I=\bigl((u_i,v_i)\bigr)_{i\in[k]}$ with $k\in\bN$ and
  $u_i,v_i\in \{a,b\}^*$ for all $i\in[k]$. A solution is a nonempty
  tuple $(i_1,\cdots,i_n)$ with entries in $[k]$ such that
  \[
    u_{i_1}\,u_{i_2}\,\cdots\,u_{i_n}
    = v_{i_1}\,v_{i_2}\,\cdots\,v_{i_n}\,.
  \]
  The existence of a solution is undecidable \cite{Pos46}.
  
  Now let $I=\bigl((u_i,v_i)\bigr)_{i\in[k]}$ be an instance of Post's
  correspondence problem.  From this finite tuple of words, we
  construct the following tPDS $\cP$. It has four states $\iota,p,q,r$
  and the following transitions:
  \begin{itemize}
  \item
    $\iota\xrightarrow{\top\mid u_i\,f(v_i)\top}p\xrightarrow{\top\mid
      u_i\,f(v_i)\top}p\xrightarrow{\top\mid\top}r$ for all $i\in[k]$
  \item
    $\iota\xrightarrow{\top\mid x\,f(x)\top}q\xrightarrow{\top\mid
      x\,f(x)\top}q\xrightarrow{\top\mid\top}r$ for all $x\in\{a,b\}$.
  \end{itemize}
  We claim that the instance $I$ has a solution if, and only if, there
  are configurations $(p,t_1)$, $(q,t_2)$, and $(r,t_3)$ such that
  \[
    (\iota,\top)\vdash^*(p,t_1)\vdash^*(r,t_3)\text{ and }
    (\iota,\top)\vdash^*(q,t_2)\vdash^*(r,t_3)\,.
  \]
  In other words, iff $\ConfGraph(\cP)$ satisfies the formula
  \[
    \varphi=\exists c_1,c_2,c_3\colon
    \begin{array}[t]{ll}
      & state_p(c_1)\land\state_q(c_2)\land\state_r(c_3)\\
      \land &(\iota,[\top])\vdash^* c_1 \vdash^* c_3\\
      \land &(\iota,[\top])\vdash^* c_2 \vdash^* c_3\,.
    \end{array}
  \]
  The formula $\varphi$ expresses the existence of a solution since
  $(\iota,\top)\vdash^*(p,t_1)\vdash^*(r,t_3)$ is equivalent to the
  existence of $1\le n\le N$ and $i_1,\dots,i_N\in[k]$ such that
  \begin{align*}
    t_1&=[u_{i_1}\,f(v_{i_1})\,u_{i_2}\,f(v_{i_2})\cdots u_{i_n}\,f(v_{i_n}) \top]
    \text{ and }\\
    t_3&=[u_{i_1}\,f(v_{i_1})\,u_{i_2}\,f(v_{i_2})\cdots u_{i_N}\,f(v_{i_N})\top]\,.
  \end{align*}
  Since $u_i\parallel f(v_j)$ for all $i,j\in[k]$, we have
  \begin{align*}
    t_3&=[u_{i_1}\,f(v_{i_1})\,u_{i_2}\,f(v_{i_2})\cdots u_{i_N}\,f(v_{i_N})\top]\\
    &=[u_{i_1}\,u_{i_2}\cdots u_{i_N}\,f(v_{i_1})\,f(v_{i_2})\cdots f(v_{i_N})\top]\\
    &=[u_{i_1}\,u_{i_2}\cdots u_{i_N}\,f(v_{i_1}\,v_{i_2}\cdots v_{i_N})\top]\,.
  \end{align*}
  Similarly, $(\iota,\top)\vdash^*(q,t_2)\vdash^*(r,t_3)$ is
  equivalent to the existence of $1\le m\le M$ and
  $x_1,x_2,\dots,x_M\in\{a,b\}$ such that
  \begin{align*}
    t_2&=[x_1\,f(x_1)\,x_2\,f(x_2)\cdots x_m\,f(x_m) \top]
    \text{ and }\\
    t_3&=[x_1\,f(x_1)\,x_2\,f(x_2)\cdots x_M\,f(x_M)\top]\\
    &=[x_1x_2\cdots x_M\,f(x_1x_2\cdots x_M)\top]\,.
  \end{align*}
  Thus, the formula $\varphi$ holds if, and only if, there exists a
  word $w\in\{a,b\}^*$ with
  \[
    u_{i_1}\,u_{i_2}\cdots u_{i_N}\,f(v_{i_1}\,v_{i_2}\cdots v_{i_N})
    = w\,f(w)\,.
  \]
  Since the function $f$ is bijective, this is equivalent to
  \[
    u_{i_1}\,u_{i_2}\cdots u_{i_N}=v_{i_1}\,v_{i_2}\cdots v_{i_N}\,,
  \]
  i.e., to the existence of a solution for $I$.    \endproof
\end{example}
In order to rule out this undecidability, we restrict the class of
pushdown systems further.

\begin{definition}\label{Def-loop-connected}
  Let $\cP=(Q,\Sigma,D,\Delta)$ be a tPDS.
  \begin{enumerate}
  \item The tPDS $\cP$ is \emph{saturated} if (for any $p,q,r\in Q$,
    $a,b\in\Sigma$, and $u,v\in\Sigma^*$)
    \[
      p\xrightarrow{a\mid ubv}q\xrightarrow{b\mid\varepsilon}r
      \text{ and }b\parallel v
      \text{ imply }p\xrightarrow{a\mid uv}r\,.
    \]

  \item The tPDS $\cP$ is \emph{loop-connected} if it is saturated and
    the following holds (for all $n\in\bN$, $p_i\in Q$,
    $a_i\in \Sigma$, $u_i,v_i\in\Sigma^*$ for $i\in[n]$). If
    \[
      p_0\xrightarrow{a_0\mid u_1a_1v_1} p_1 \xrightarrow{a_1\mid u_2a_2v_2}
      p_2 \xrightarrow{a_2\mid u_3a_3v_3}\cdots \xrightarrow{a_{n-1}\mid
        u_n a_n v_n} p_n
    \]
    with $a_i\parallel v_i$ for all $i\in[n]$ and
    $(p_0,a_0)=(p_n,a_n)$, then the set $\alphabet(u_1u_2\cdots u_n)$
    is connected.
  \end{enumerate}
\end{definition}
Suppose we have
$p\xrightarrow{a\mid ubv}q\xrightarrow{b\mid\varepsilon}r$ and
$b\parallel v$. Then, for any trace $x$, one gets
\[
  (p,x\cdot[a])\vdash(q,x\cdot[ubv])=(q,x\cdot[uvb])\vdash(r,x\cdot[uv])\,.
\]
If $\cP$ is saturated, we can take the ``shortcut''
$(p,x\cdot[a])\vdash(r,x\cdot[uv])$ in a single step. From an
arbitrary tPDS, one can compute a saturated tPDS by adding transitions
that has the same reachability relation \cite[Pro.~5.18]{KoeK26}. It therefore
makes sense to only consider saturated systems.

The idea of loop-connectedness is as follows. Suppose there are
transitions as in the definition. Since $\cP$ is a tPDS that, in
certain transitions, can replace $a_i$ by a word containing $a_{i+1}$,
we get
\[
  D(a_n)\subseteq D(a_{n-1})\subseteq D(a_{n-2})\subseteq\cdots
  \subseteq D(a_0)=D(a_n)
\]
implying $D(a_i)=D(a_0)$ for all $i\in[n]$.  Next let
$i\in\{0,\dots,n-1\}$ and suppose $b$ is a letter that occurs in
$v_{i+1}$. Since the tPDS $\cP$ can replace $a_i$ with
$u_{i+1}a_{i+1}v_{i+1}$, we get
$b\in D(u_{i+1}a_{i+1}v_{i+1})\subseteq D(a_i)=D(a_{i+1})$ and
therefore $(b,a_{i+1})\in D$. As this contradicts
$a_{i+1}\parallel v_{i+1}$, the letter $b$ cannot exist, i.e.,
$v_{i+1}$ is the empty word~$\varepsilon$. Consequently, the
transitions from the definition are actually of the form
\[
  p_0\xrightarrow{a_0\mid u_1a_1} p_1 \xrightarrow{a_1\mid u_2a_2}
  p_2 \xrightarrow{a_2\mid u_3a_3}\cdots \xrightarrow{a_{n-1}\mid
    u_n a_n} p_n
\]
with $(p_0,a_0)=(p_n,a_n)$.  Then the tPDS can make the following
computation
\begin{align*}
  (p_0,[a_0]) & \vdash (p_1,[u_1a_1])\\
              & \vdash (p_2,[u_1u_2a_2])\\
              &\vdots\\
     & \vdash (p_n,[u_1u_2\cdots u_n a_n])\\
     & = (p_0,[u_1u_2\cdots u_n a_0])\,.
\end{align*}
In other words, the tPDS can replace any trace $[u\,a_0]$ by the trace
$[u\,u_1\cdots u_n\,a_0]$ and therefore, by repeating this loop, by
$[u\,(u_1\cdots u_n)^m\,a_0]$ for any number $m$. The requirement is
that the alphabet of the looped word $u_1u_2\cdots u_n$ shall be
connected.

As a side-remark, we indicate that loop-connectedness is difficult to
verify.

\begin{proposition}
  The following is coNP-hard.

  input: saturated tPDS $\cP=(Q,\Sigma,D,\Delta)$

  question: is $\cP$ loop-connected
\end{proposition}

\begin{proof}
  Let $(\Sigma,D)$ be a dependence alphabet and $\cA=(Q,\Sigma,I,T,F)$
  an nfa. From $(\Sigma,D)$ and $\cA$, we construct the pushdown
  system $\cP=(Q,\Gamma,D',\Delta)$ setting
  \begin{itemize}
  \item $\Gamma=\Sigma\cup\{\top\}$,
  \item $D'=D\cup (\Sigma\times\{\top\})\cup(\{\top\}\times\Sigma)$, and
  \item $\Delta=\{(p,\top,a\top,q)\mid (p,a,q)\in T\}$.
  \end{itemize}
  Then it is easily checked that $\cP$ is a saturated tPDS since none
  of its transitions writes $\varepsilon$ onto the
  pushdown. Furthermore, it is loop-connected if, and only if, for
  every loop
  \[
    (p_0,a_1,p_1)(p_1,a_2,p_2)\cdots (p_{n-1},a_n,p_n)
  \]
  in $\cA$ with $p_n=p_0$, the set
  $\{a_1,a_2\cdots,a_n\}\subseteq\Sigma$ is connected.

  The set of nfas with this property is coNP-complete
  \cite[Prop.~4]{MusP99}, thus we reduced from a coNP-complete problem.
\end{proof}

We can now formulate the main result of this paper.
\begin{theorem}\label{T-main}
  The following problem is decidable.\\
  input:
  \begin{minipage}[t]{.7\linewidth}
    loop-connected tPDS $\cP$\\
    first-order formula $\varphi$ in the language of $\ConfGraph(\cP)$
  \end{minipage}\\
  question: Does $\varphi$ hold in $\ConfGraph(\cP)$?
\end{theorem}
Before starting the proof, we demonstrate that this decidability does
not hold for monadic second-order logic.

\begin{example}\label{E-MSO}
  Let $\cP=(Q,\Sigma,D,\Delta)$ be the pushdown system with $Q=\{q\}$,
  $\Sigma=\{a,b\}$, $(a,b)\in I$, and $(q,a,aa,q),(q,b,bb,q)$ the only
  transitions. Since any transition writes the letter it reads, $\cP$
  is a tPDS. Since no transition writes $\varepsilon$, it is
  saturated. Note that in any loop as in the
  Definition~\ref{Def-loop-connected}, all letters $a_i$ and words
  $u_i,v_i$ belong to $a^+$ or all of them belong to $b^+$. Hence,
  $\cP$ is even loop-connected.

  We have $(q,[a^kb^\ell])\vdash^*(q,[a^mb^n])$ if, and only if,
  $k\le m$ and $\ell\le n$. Hence $\ConfGraph(\cP)$ contains the
  infinite grid $(\bN\times\bN,\le)$ implying that the monadic
  second-order theory is undecidable (cf., e.g.,
  \cite[Thm.~5.6]{CouE12}).
\end{example}

To prove Theorem~\ref{T-main}, we work with the following structure.
\begin{definition}
  Let $\cP=(Q,\Sigma,D,\Delta)$ be a pushdown system.  For $p,q\in Q$,
  we let $\reach_{p,q}\subseteq\bM^2$ denote the set of pairs $(s,t)$
  of traces such that $(p,s)\vdash^*(q,t)$. Furthermore,
  $\step_{p,q}\subseteq\bM^2$ is the set of pairs $(s,t)$ such that
  $(p,s)\vdash(q,t)$.

  Define the structure
  \[
    \cS(\cP)=\bigl(\bM,(\step_{p,q},\reach_{p,q})_{p,q\in Q},(t)_{t\in\bM}\bigr)
  \]
  with universe the set of traces over $(\Sigma,D)$, the one-step
  relations $\step_{p,q}$ and the reachability relations
  $\reach_{p,q}$ for all $p,q\in Q$, and all traces $t\in\bM$ as
  constants.
\end{definition}

To justify that we work with this structure, we first prove that the
theory of $\ConfGraph(\cP)$ can be reduced to that of~$\cS(\cP)$.
This proof can be understood as interpretation of the configuration
graph $\ConfGraph(\cP)$ in the structure $\cS(\cP)$. We refrain from
defining this notion in full generality and listing its consequences
(cf.\ \cite[p.~212]{Hod93}) as this would be more complicate than the
direct proof.

\begin{lemma}\label{L-reduction}
  From a pushdown system $\cP=(Q,\Sigma,D,\Delta)$, a first-order
  formula $\varphi(x_1,\dots,x_n)$ in the language of
  $\ConfGraph(\cP)$, and a tuple of states
  $\ol{p}=(p_1,\dots,p_n)\in Q^n$, one can compute a formula
  $\varphi_{\ol{p}}(x_1,\dots,x_n)$ in the language of $\cS(\cP)$ such
  that
  \[
    \ConfGraph(\cP)\models\varphi((p_i,t_i))_{i\in[n]})
    \iff
    \cS(\cP)\models\varphi_{\ol{p}}(t_1,\dots,t_n)
  \]
  holds for any $t_1,\dots,t_n\in\bM$.
\end{lemma}
\begin{proof}
  The proof is by induction on the construction of $\varphi$.
  \begin{itemize}
  \item If $\varphi=(x_i=x_j)$ and $p_i=p_j$, then
    $\varphi_{\ol{p}}=(x_i=x_j)$.
  \item If $\varphi=(x_i=x_j)$ and $p_i\neq p_j$, then
    $\varphi_{\ol{p}}=(1\neq 1)$.
  \item If $\varphi=(x_i=(p,t))$ and $p_i= p$, then
    $\varphi_{\ol{p}}=(x_i= t)$.
  \item If $\varphi=(x_i=(p,t))$ and $p_i\neq p$, then
    $\varphi_{\ol{p}}=(1\neq 1)$.
  \item If $\varphi=\state_p(x_i)$ and $p_i=p$, then
    $\varphi_{\ol{p}}=(1=1)$.
  \item If $\varphi=\state_p(x_i)$ and $p_i\neq p$, then
    $\varphi_{\ol{p}}=(1\neq 1)$.
  \item If $\varphi=(x_i\vdash x_j)$, then
    $\varphi_{\ol{p}}=\step_{p_i,p_j}(x_i,x_j)$.
  \item If $\varphi=(x_i\vdash^* x_j)$, then
    $\varphi_{\ol{p}}=\reach_{p_i,p_j}(x_i,x_j)$.
  \item If $\varphi=\alpha\lor\beta$, then
    $\varphi_{\ol{p}}=\alpha_{\ol{p}}\lor\beta_{\ol{p}}$.
  \item If $\varphi=\lnot\alpha$, then
    $\varphi_{\ol{p}}=\lnot\alpha_{\ol{p}}$.
  \item If $\varphi=\exists x_{n+1}\colon\alpha$, then
    $\varphi_{\ol{p}}=\displaystyle\bigvee_{p_{n+1}\in Q}\exists
    x_{n+1}\colon\alpha_{\ol{p},p_{n+1}}$. \vspace*{-2ex}
  \end{itemize}
\end{proof}

\begin{remark}
  The above lemma can also be shown for monadic second-order logic. In
  addition to the above translation, one defines:
  \begin{itemize}
  \item If $\varphi=\exists X\,\alpha$, then
    $\varphi_{\ol{p}}=\exists (X_p)_{p\in Q}\,\alpha_{\ol{p}}$.
  \item If $\varphi=(x_i\in X)$, then
    $\varphi_{\ol{p}}=(x_i\in X_{p_i})$.
  \end{itemize}
\end{remark}

If the theory of $\ConfGraph(\cP)$ is undecidable, then, by the above
lemma, the same applies to the theory of $\cS(\cP)$. This implies in
particular the following negative results.
\begin{itemize}
\item The first-order theory of $\cS(\cP)$ is not uniformly decidable
  for all tPDS (cf.~Example~\ref{E-undecidable}).
\item The monadic second-order theory of $\cS(\cP)$ is not uniformly
  decidable for all loop-connected tPDS (cf.~Example~\ref{E-MSO}).
\end{itemize}
Conversely, in order to prove Theorem~\ref{T-main}, it suffices to
demonstrate the following result.
\begin{theorem}\label{T-semi-main}
  The following problem is decidable.\\
  input:
  \begin{minipage}[t]{.7\linewidth}
    loop-connected tPDS $\cP$\\
    first-order formula $\varphi$ in the language of $\cS(\cP)$
  \end{minipage}\\
  question: Does $\varphi$ hold in $\cS(\cP)$?
\end{theorem}
The proof of this theorem can be found in Section~\ref{S-proof}.  It
is based on a special class of trace relations
(cf.~Definition~\ref{D-fnf-automatic} below) that we introduce and
study next. Later, it will turn out that the relations $\step_{p,q}$
and $\reach_{p,q}$ of a loop-connected tPDS belong to this class which
implies the decidability using the theory of automatic structures.

\section{fnf-automatic relations}

Let $A=\{a_1,\dots,a_n\}\subseteq\Sigma$.  The set $A$ is
\emph{independent} if $(a_i,a_j)\in I$ for all $1\le i<j\le n$; let
$\cF\subseteq 2^\Sigma$ denote the set of all independent sets. Note
that in particular $\emptyset$ is independent. It follows that, for
any permutation $\sigma\colon[n]\to[n]$, we have
\[
  a_1a_2\cdots a_n \sim a_{\sigma(1)} a_{\sigma(2)}\cdots a_{\sigma(n)}\,.
\]
Hence there is a unique trace $[A]=[a_1\cdots a_n]$. This mapping
$[.]\colon\cF\to\bM$  extends uniquely to a monoid homomorphism
$[.]\colon\cF^*\to\bM$.

A (possibly empty) word $A_1\,A_2\,\cdots\,A_n\in\cF^*$ is \emph{in
  extended Foata normal form} if $D(A_i)\supseteq A_{i+1}$ for all
$i\in[n-1]$. In other words, we have
\begin{center}
  for all $i\in[n-1]$ and $b\in A_{i+1}$, there exists $a\in A_i$ with
  $(a,b)\in D$.
\end{center}

\begin{example}\label{E-eFNF}
  Suppose $\Sigma=\{a,b,c\}$ with $D$ the reflexive and symmetric
  closure of $\{(a,b),(b,c)\}$ such that $(a,c)\in I$. Then
  $\{a,c\}$, $\{a\}$, $\{b\}$, $\{c\}$, $\emptyset$ are the only
  elements of $\cF$. Furthermore, the words
  $W=\{a,c\}\,\{b\}\,\{a,c\}$ and $\{a\}\,\{b\}\{a,c\}$ are in
  extended Foata normal form (even in Foata normal form, see
  below for the definition). Note that also the words
  $W\emptyset\,\emptyset$ and $\{a,c\}\,\{a,c\}\,\cdots\,\{a,c\}$
  are in extended Foata normal form.

  On the other hand, $\{a\}\,\{a,c\}$ is not in extended Foata normal
  form since $(a,c)\notin D$.
\end{example}

Let $W=A_1\,A_2\,\cdots\,A_n$ be in extended Foata normal form. Since
$D(\emptyset)=\emptyset$, there is $m\in\{0,\dots,n\}$ such that
$A_i=\emptyset\iff i>m$ for all $i\in[n]$, i.e., in a word in extended
Foata normal form, the letter $\emptyset$ can occur at the end of the
word, only.

A word over $\cF$ is \emph{in Foata normal form} if it is in extended
Foata normal form and none of its letters equals $\emptyset$. By the
very definition, any infix of any word in (extended) Foata normal form
is in (extended) Foata normal form, again.

Let $x\in\bM$. Then, by \cite[Theorem~1.2]{CarF69}, there is a unique
word $\fnf(x)\in\cF^*$ in Foata normal form with $x=[\fnf(x)]$; this
word is the \emph{Foata normal form of $x$}. In the situation of
Example~\ref{E-eFNF} with $(a,c)\in I$, $\fnf([a^nc^n])$ equals
\[
  W=\underbrace{\{a,c\}\,\{a,c\}\,\cdots\,\{a,c\}}_{n\text{ times}}
\]
since 
$W$ is in Foata normal form, $[\{a,c\}]$ is the trace $[ac]$ and
$[ac]^n=[(ac)^n]=[a^nc^n]$.

Suppose $W=A_1\,\dots\,A_n$ is in extended Foata normal form with
$x=[W]$. Then there exists $m\in\{0,\dots,n\}$ such that
$\fnf(x)=A_1\,A_2\,\cdots\,A_m$, i.e., the Foata normal form of $x$ is
a prefix of the word $W$. It follows that the set of words $W$ in
extended Foata normal form with $x=[W]$ equals $\fnf(x)\,\{\emptyset\}^*$.

Let $W=A_1\,A_2\,\cdots\,A_n$ be in Foata normal form. Then
any factorisation of $W$ into $U$ and $V$ trivially satisfies
$[W]=[U]\cdot[V]$; furthermore, $U$ and $V$ are the Foata normal forms
of the two factors $[U]$ and $[V]$, respectively. But there are
factorisations of $[W]$ into traces $x$ and $y$ such that no
factorsation of $W$ into $U$ and $V$ satisfies $x=[U]$ and $y=[V]$; in
particular, the Foata normal forms of $x$ and $y$ need not be factors
of the Foata normal form of $x\cdot y$. The following lemma describes
the factorisations of the trace $[W]$ in terms of the word $W$.

\begin{lemma}\label{L-fnf-product}
  Let $W=A_1\,\cdots\,A_n$ be a word in extended Foata normal form and
  $x,y\in\bM$. Then the following are equivalent:
  \begin{enumerate}
  \item[(1)] $x\cdot y=[W]$
  \item[(2)] There are sets $B_i\subseteq A_i$ for $i\in[n]$ such that
    \begin{enumerate}
    \item the word
      $(A_1\setminus B_1)\,(A_2\setminus B_2)\,\cdots\,(A_n\setminus
      B_n)$ is in extended Foata normal form,
    \item $B_i\parallel (A_j\setminus B_j)$ for all $1\le i<j\le n$, and
    \item
      $x=\bigl[(A_1\setminus B_1)\,(A_2\setminus
      B_2)\,\cdots\,(A_n\setminus B_n)\bigr]$ and
      $y=[B_1\,B_2\,\cdots\,B_n]$.
    \end{enumerate}
  \end{enumerate}
\end{lemma}

Note that, by (a) and (c), the Foata normal form of $x$ is a prefix of
the word
$(A_1\setminus B_1)\,(A_2\setminus B_2)\,\cdots\,(A_n\setminus
B_n)$. In contrast, the Foata normal form of $y$ can be very different
from the word $B_1\,B_2\,\cdots\,B_n$: let
$\Sigma=\{a_i,b_i\mid i\in[3]\}$ and $D$ the reflexive and transitive
closure of
\[
  \{(a_i,a_j)\mid i,j\in[3]\}\cup\{(a_i,b_j)\mid 1\le i<j\le 3\}\,.
\]
Then $A_i=\{a_i,b_i\}$ is an element of $\cF$ and the word
$W=A_1A_2A_3$ is in Foata normal form since
$a_i\in D(a_{i+1})\cap D(b_{i+1})$ for all $i\in[2]$. With
$B_i=\{b_i\}$ for all $i\in[3]$, properties (a) and (b) hold. Note
that the Foata normal form of the trace $y=[B_1B_2B_3]$ equals
$\{b_1,b_2,b_3\}$ since these three letters are mutually independent.

\newenvironment{proofof}
  {\trivlist\PRstyle\item[]{\bfseries Proof of Lemma~\ref{L-fnf-product}:}\newline}{\QED\endtrivlist}

\begin{proofof}
  First, assume (2) holds.  From (b), we obtain
  $[A_j\setminus B_j]\cdot[B_1]=[B_1]\cdot[A_j\setminus B_j]$ 
  for all $j\in\{2,\dots,n\}$. Hence
  \begin{align*}
    x\cdot [B_1] &=[A_1\setminus B_1]\cdots[A_n\setminus B_n]
                   \cdot [B_1]\\
                 &=[A_1\setminus B_1]\cdot[B_1]\cdot[A_2\setminus B_2]\cdots[A_n\setminus B_n]\\
                 &=[A_1]\cdot[A_2\setminus B_2]\cdots[A_n\setminus B_n]\,.
  \end{align*}
  Similarly,
  $[A_j\setminus B_j]\cdot[B_2]=[B_2]\cdot[A_j\setminus B_j]$ holds
  for all $j\in\{3,\dots,n\}$ implying
  \begin{align*}
    x\cdot[B_1]\cdot[B_2] &=[A_1]\cdot[A_2]\cdot[A_3\setminus B_3]\cdots[A_n\setminus B_n]\,.
  \end{align*}
  Inductively, we get
  \begin{align*}
    x\cdot y &= [A_1\setminus B_1]\cdots[A_n\setminus B_n]
               \cdot[B_1]\cdots[B_n]\\
    &= [A_1]\cdot[A_2]\cdots [A_n]=[W]
  \end{align*}
  and therefore (1).

  Conversely, suppose (1) holds.  Levi's Lemma
  (cf.~\cite[Prop.~3.2.3]{DieR95}) explicitly says that there are
  traces $x_i$ and $y_i$ for all $i\in[n]$ such that
  \begin{itemize}
  \item $[A_i]=x_i\cdot y_i$ for all $i\in[n]$,
  \item $y_i\parallel x_j$ for all $1\le i<j\le n$,
  \item $x=x_1\cdot x_2\cdots x_n$, and
  \item $y=y_1\cdot y_2\cdots y_n$.
  \end{itemize}
  Since, for $i\in[n]$, $A_i$ is an independent set with
  $[A_i]=x_i\cdot y_i$, there is $B_i\subseteq A_i$ with
  $x_i=[A_i\setminus B_i]$ and $y_i=[B_i]$.

  We first show (a), i.e.,
  $D(A_i\setminus B_i)\supseteq A_{i+1}\setminus B_{i+1}$ for all
  $i\in[n-1]$.

  So let $i\in[n-1]$ and
  $b\in A_{i+1}\setminus B_{i+1}\subseteq A_{i+1}$. Since the word
  $A_1\,A_2\,\cdots\,A_n$ is in extended Foata normal form, there
  exists $a\in A_i$ with $(a,b)\in D$. Suppose $a\in B_i$. Then we
  have $a\in B_i=\alphabet(y_i)$ and
  $b\in A_{i+1}\setminus B_{i+1}=\alphabet(x_{i+1})$. From
  $y_i\parallel x_{i+1}$, we obtain $(a,b)\in I$, which contradicts
  $(a,b)\in D$. Thus, indeed, (a) holds.

  To verify claim (b), let $1\le i<j\le n$. Then
  $B_i =\alphabet(y_i) \parallel \alphabet(x_j)= A_j\setminus B_j$.

  Claim (c) is obvious by
  \begin{align*}
    x &= x_1\cdot x_2\cdots x_n\\
      &= [A_1\setminus B_1]\cdot[A_2\setminus B_2]\cdots[A_n\setminus B_n]
      &&\text{since } x_i=[A_i\setminus B_i]\text{ for all }i\in[n]
         \intertext{ and }
    y &= y_1\cdot y_2\cdots y_n\\
      &=[B_1]\cdot[B_2]\cdots[B_n]
      &&\text{since }y_i=[B_i]\text{ for all }i\in[n]\,. 
  \end{align*}
\end{proofof}

We next define what we mean by ``fnf-automatic relations''.

Let $\cR\subseteq\bM^k$ be a $k$-ary relation on the set of traces
$\bM$. From this relation, we first construct a language $L_\cR$ over
the alphabet $\cF^k$. Note that the elements of $\cF^k$ are $k$-tuples
of independent sets that we write interchangeably as
\[
  (A_1,\dots,A_k)\quad \text{ or }\quad
  \begin{pmatrix}
    A_1\\A_2\\\vdots\\A_k
  \end{pmatrix}\,.
\]
Now consider a word
\[
  \begin{pmatrix}
    A_1^1\\A_2^1\\\vdots\\A_k^1
  \end{pmatrix}
  \begin{pmatrix}
    A_1^2\\A_2^2\\\vdots\\ A_k^2
  \end{pmatrix}
  \cdots
  \begin{pmatrix}
    A_1^n\\A_2^n\\\vdots \\A_k^n
  \end{pmatrix}
  \in(\cF^k)^*
\]
over $\cF^k$. It belongs to the language $L_\cR$ if
\begin{itemize}
\item $W_i=A_i^1\,A_i^2\,\cdots\,A_i^n$ is in extended Foata normal form
  for all $i\in[k]$ and
\item the tuple $([W_1],[W_2],\dots,[W_k])$ belongs to $\cR$.
\end{itemize}
Intuitively, we understand a word $W$ over $\cF^k$ as the tuple of words
$(W_1,\dots,W_k)$ (of the same length) over $\cF$. Then $W$ belongs to
$L_\cR$ if all the words $W_i$ are in extended Foata normal form and
the tuple $(x_1,\dots,x_k)$ of traces represented by these words
belongs to the relation $\cR$.

As an example, suppose $A,B\in\cF$ and $\cR=\{([A][A],[B])\}$. Then $L_\cR$
is the language
\[
  \begin{pmatrix}
    A\\B
  \end{pmatrix}
  \begin{pmatrix}
    A\\\emptyset
  \end{pmatrix}
  \begin{pmatrix}
    \emptyset\\\emptyset
  \end{pmatrix}^*\,.
\]

\begin{definition}\label{D-fnf-automatic}
  A relation $\cR\subseteq\bM^k$ is \emph{fnf-automatic} if the
  language $L_\cR\subseteq(\cF^k)^*$ is regular.
\end{definition}

By the very definition, any fnf-automatic relation $\cR$ can be
represented by an NFA $\cA$ for the language $L_\cR$; we will always
assume that fnf-automatic relations are presented that way.

Since the set of all words in extended Foata normal form is regular,
the identity relation $\Id_\bM=\{(x,x)\mid x\in\bM\}$ is
fnf-automatic. It is not difficult to see that the union and the
intersection of two fnf-automatic relations is fnf-automatic. The same
applies to the complement as well as the inversion
$\cR^{-1}=\{(x_k,x_{k-1},\dots,x_1)\mid (x_1,x_2,\dots,x_k)\in\cR\}$.
The following lemma shows that the class of fnf-automatic relations is
closed under composition, the proof is simple, but not obvious.

\begin{lemma}
  If $\cR_1$ and $\cR_2$ are fnf-automatic relations, then
  $\cR_1\circ\cR_2$ is fnf-automatic.
\end{lemma}

\begin{proof}
  For notational convenience, we assume $\cR_1$ and $\cR_2$ to be
  binary. For $1\le i<j\le 3$, let
  $\pi_{ij}\colon(\cF^3)^*\to(\cF^2)^*$ be the monoid homomorphism with
  \[
    \begin{pmatrix}
      A_1\\A_2\\A_3
    \end{pmatrix}
    \mapsto\pair{A_i}{A_j}
  \]
  for all $A_1,A_2,A_3\in\cF$. We consider the language
  \[
    L=\pi_{13}\bigl(\pi_{12}^{-1}(L_{\cR_1})\cap\pi_{23}^{-1}(L_{\cR_2})\bigr)\,.
  \]

  For any words $U=A_1\,A_2\,\cdots\,A_n$ and
  $W=C_1\,C_2\,\cdots\,C_n$ in extended Foata normal form, we have
  \[
    \pair{A_1}{C_1}\,\pair{A_2}{C_2}\,\cdots\,\pair{A_n}{C_n}\in
    L_{\cR_1\circ\cR_2}
  \]
  if, and only if, there exists $m\in\bN$ with
  \[
    \pair{A_1}{C_1}\,\pair{A_2}{C_2}\,\cdots\,\pair{A_n}{C_n}\,
    \pair{\emptyset}{\emptyset}^m\in L\,.
  \]
  Hence $L_{\cR_1\circ\cR_2}$ is the right quotient of the language
  $L$ wrt.\ the regular language
  $\pair{\emptyset}{\emptyset}^*$. Since the class of regular
  languages is closed under homomorphic images and inverse images and
  under right quotions wrt.\ regular languages, the language
  $L_{\cR_1\circ\cR_2}$ is regular.  
\end{proof}

The following lemma relates the class of fnf-automatic relations to
that of rational relations, i.e., relations that are constructed from
finite relations using (componentwise) concatenation, Kleene star, and
union (see \cite{Ber79} for a comprehensive study of this class of
relations).

\begin{lemma}\label{L-fnf-vs-rational}
  Any fnf-automatic relation is rational, but there are rational unary
  relations (i.e., rational languages) that are not fnf-automatic.
\end{lemma}

\begin{proof}
  For the first claim, let $\cR\subseteq\bM^k$ be fnf-automatic. Then
  the language $L_\cR$ is regular and $\cR$ is the image of $L_\cR$
  under the monoid homomorphism
  \[
    (\cF^k)^*\to\bM^k\colon
    \begin{pmatrix}
      A_1\\A_2\\\vdots\\A_k
    \end{pmatrix}
    \mapsto
    \bigl([A_1],[A_2],\dots,[A_k]\bigr)\,.
  \]
  Since homomorphic images of rational sets are rational, the relation
  $\cR$ is rational.

  For the second claim, let $a,b\in \Sigma$ with $a\parallel b$ and
  consider the rational set of traces
  $\cL=[aab]^*=\{[aab]^n\mid n\in\bN\}$. Then $L_\cL$ is the set of
  words $A^n\,B^n\,C^m$ with $A=\{a,b\}$, $B=\{a\}$, $C=\emptyset$,
  and $m,n\in\bN$. Since this language is not regular, the unary
  relation $\cL$ is not fnf-automatic.
\end{proof}

In general, the concatenation of two fnf-automatic relations is not
fnf-automatic (consider, e.g., the unary relations $\cL_1=[a]^*$ and
$\cL_2=[bc]^*$ with $(a,b)\in D$ and $\{a,b\}\parallel\{c\}$ since
$[(ac)^m\,b^n]\in\cL_1\cdot\cL_2$ iff $m=n$). The following lemma
proves that certain concatenations are fnf-automatic.

\begin{lemma}\label{L-fnf-sr}
  Let $\cR\subseteq\bM^2$ be fnf-automatic and $\cK,\cL\subseteq\bM$
  recognizable. Then
  \[
    \cR\cdot(\cK\times\cL)=\{(x_1y,x_2z)\mid (x_1,x_2)\in\cR, y\in\cK, z\in\cL\}
  \]
  is effectively fnf-automatic.
\end{lemma}

\begin{proof}
  Note that
  \[
    \cR\cdot(\cK\times\cL)=\bigl(\cR\cdot(\cK\times\{1\})\bigr)
    \circ
    \bigl(\Id_\bM\cdot(\cL\times\{1\})\bigr)^{-1}\,.
  \]
  Since the class of fnf-automatic relations is effectively closed
  under composition and inversion and since $\Id_\bM$ is
  fnf-automatic, it suffices to prove the lemma for $\cL=\{1\}$.

  The idea of the construction of an automaton for
  $\cR\cdot(\cK\times\{1\})$ is as follows. It reads a word
  \[
    \pair{A_1}{C_1}
    \pair{A_2}{C_2}
    \cdots
    \pair{A_n}{C_n}
  \]
  with $A_1\,A_2\,\cdots\,A_n$ and $C_1\,C_2\,\cdots\,C_n$ in extended
  Foata normal form. Along this path, it guesses sets
  $B_i\subseteq A_i$ with the following properties:
  \begin{itemize}
  \item The word
    $(A_1\setminus B_1)(A_2\setminus B_2)\cdots(A_n\setminus B_n)$ is
    in extended Foata normal form.

    To ensure this, we recall the last set $A_i\setminus B_i$ and
    enforce that any letter from $A_{i+1}\setminus B_{i+1}$ is
    dependent from some letter in the recalled set.
  \item The word
    \[
      \pair{A_1\setminus B_1}{C_1}
      \pair{A_2\setminus B_2}{C_2}
      \cdots
      \pair{A_n\setminus B_n}{C_n}
    \]
    represents some pair in the relation $\cR$.

    This is realized by running an automaton for $L_\cR$ that consumes
    the pairs $\pair{A_i\setminus B_i}{C_i}$.
  \item The word $B_1\,B_2\,\cdots\,B_n$ represents some trace from
    $\cK$.

    For this, we run an automaton for $\cK$ that consumes the sets
    $B_i$.
  \item The word $A_1\,A_2\,\cdots\,A_n$ represents the product of the
    traces
    $\bigl[(A_1\setminus B_1)\,(A_2\setminus
    B_2)\,\cdots\,(A_n\setminus B_n)\bigr]$ and
    $[B_1\,B_2\,\cdots\,B_n]$.

    This is achieved by recalling the set of letters that have been
    guessed so far (i.e., $\bigcup_{j\in[i]}B_j$) and enforcing that
    they are independent from all letters in $A_{i+1}\setminus B_{i+1}$.
  \end{itemize}

  We now come to the details of this construction. Since $\cR$ is
  fnf-automatic, there is an NFA $\cA_1=(Q_1,\cF^2,I_1,T_1,F_1)$ that
  accepts the language of all words
  \[
    \pair{A'_1}{C_1}
    \pair{A'_2}{C_2}
    \cdots
    \pair{A'_n}{C_n}
  \]
  such that
  \begin{itemize}
  \item $A'_1\,A'_2\,\cdots\,A'_n$ and $C_1\,C_2\,\cdots\,C_n$ are in
    extended Foata normal form and
  \item $\bigl([A'_1\,A'_2\,\cdots\,A'_n],[C_1\,C_2\,\cdots\,C_n]\bigr)\in\cR$.
  \end{itemize}
  Since $\cK\subseteq\bM$ is recognizable, there exists an NFA
  $\cA_2=(Q_2,\Sigma,I_2,T_2,F_2)$ that accepts the language
  \[
    L(\cA_2)=\{w\in\Sigma^*\mid [w]\in\cK\}\,.
  \]

  We now build a new NFA
  $\cA=(Q,\cF^2,I,T,F)$ as follows:
  \begin{itemize}
  \item The set of states equals
    $Q=2^\Sigma\times Q_1\times Q_2\times 2^\Sigma$, i.e., states of
    $\cA$ are quadruples consisting of states of $\cA_1$ and $\cA_2$
    as well as two sets of letters.
  \item The set of initial states equals
    $I=\{\Sigma\}\times I_1\times I_2\times\{\emptyset\}$, i.e., the
    set of quadruples $(\Sigma,\iota_1,\iota_2,\emptyset)$ with
    $\iota_i\in I_i$.
  \item We have
    $\left((X_1,p_1,p_2,X_2),\pair{A}{C},(Y_1,q_1,q_2,Y_2)\right)\in
    T$ if, and only if, there exists $B\subseteq A$ with
    \begin{itemize}
    \item $D(X_1)\supseteq A\setminus B$ and $Y_1=A\setminus B$,

    \item $\left(p_1,\pair{A\setminus B}{C}, q_1\right)\in T_1$,
    \item there is some word $w\in[B]$ with
      $p_2\xrightarrow{w}_2 q_2$, and
    \item $X_2\parallel A\setminus B$ and $Y_2=X_2\cup B$.
    \end{itemize}
  \item The set of accepting states equals
    $F=2^\Sigma\times F_1\times F_2\times 2^\Sigma$, i.e., the set of
    quadruples $(Y_1,f_1,f_2,Y_2)$ with $f_i\in F_i$.
  \end{itemize}

  By induction, one then obtains the following for any $n\in\bN$. Let
  $(\iota_1,\iota_2)\in I_1\times I_2$, $A_1\,A_2\,\cdots\,A_n$ and
  $C_1\,C_2\,\cdots\,C_n$ be words in extended Foata normal form,
  $(q_1,q_2)\in Q_1\times Q_2$, and $Y_1,Y_2\subseteq \Sigma$. Then
  \begin{equation}
    \label{eq:proof-of-}
    (\Sigma,\iota_1,\iota_2,\emptyset)
    \xrightarrow{\pair{A_1}{C_1}\,\cdots\,\pair{A_n}{C_n}}
    (Y_1,q_1,q_2,Y_2)
  \end{equation}
  if, and only if, there are sets $B_i\subseteq A_i$ for all $i\in[n]$
  such that
  \begin{enumerate}[(i)]
  \item
    $(A_1\setminus B_1)\,(A_2\setminus B_2)\,\cdots\,(A_n\setminus
    B_n)$ is in extended Foata normal form and $Y_1=A_n\setminus B_n$,

  \item $\iota_1\xrightarrow{\pair{A_1\setminus B_1}{C_1}\,\pair{A_2\setminus B_2}{C_2}\,\cdots\,\pair{A_n\setminus B_n}{C_n}}_1 q_1$,
  \item for all $i\in[n]$, there are words $w_i\in [B_i]$ such that
    $\iota_2\xrightarrow{w_1\,w_2\,\cdots\,w_n}_2 q_2$, and
  \item $B_i\parallel(A_j\setminus B_j)$ for all $1\le i<j\le n$ and
    $Y_2=\bigcup_{i\in[n]}B_i$.
  \end{enumerate}

  Now let $U=A_1\,A_2\,\cdots\,A_n$ and $V=C_1\,C_2\,\cdots\,C_n$ be
  in extended Foata normal form and let
  \[
    W=\pair{A_1}{C_1}\pair{A_2}{C_2}\cdots\pair{A_n}{C_n}\,.
  \]
  We have to show
  \begin{equation}
    \label{eq:claim}
    W\in L(\cA)\iff ([U],[V])\in\cR\cdot(\cK\times\{1\})\,.
  \end{equation}
  
  First, assume $W\in L(\cA)$. Then there are
  $(\iota_1,\iota_2)\in I_1\times I_2$, $(q_1,q_2)\in F_1\times F_2$,
  and $Y_1,Y_2\subseteq\Sigma$ such that \eqref{eq:proof-of-}
  holds. By the above, there are sets $B_i\subseteq A_i$ for all
  $i\in[n]$ such that (i-iv) hold.

  In the following, set
  \begin{align*}
    x &= \bigl[(A_1\setminus B_1)\,(A_2\setminus B_2)\,\cdots\,
        (A_n\setminus B_n)\bigr]\,, \\
    y &=[B_1\,B_2\,\cdots\,B_n]\,, \text{ and }\\
    z &=[C_1\,C_2\,\cdots\,C_n]=[V]\,.
  \end{align*}
  Since, by (i), the word
  $(A_1\setminus B_1)\cdots(A_n\setminus B_n)$ is in extended Foata
  normal form, (ii) and our assumption on the NFA $\cA_1$ imply
  $(x,z)\in \cR$.

  Property (iii) implies that $w_1w_2\cdots w_n$ is accepted by the NFA
  $\cA_2$. Hence $y=[w_1\cdots w_n]\in\cK$.

  Finally, Lemma~\ref{L-fnf-product} together with (i) and (iv) and
  our definition of $x$ and $y$ yield
  \[
    x\cdot y=[A_1]\cdot[A_2]\cdots[A_n]=[U]\,.
  \]

  Now $(x,z)\in \cR$ and $y\in \cK$ imply
  \[
    ([U],[V])=(x\cdot y,z)\in \cR\cdot(\cK\times\{1\})\,.
  \]
  This finishes the verification of the implication ``$\Rightarrow$''
  in \eqref{eq:claim}.

  Conversely, suppose $([U],[V])\in \cR\cdot(\cK\times\{1\})$. Set
  $z=[V]$ as above. There are traces $x$ and $y$ with $[U]=x\cdot y$, $(x,z)\in \cR$,
  and $y\in\cK$. Since $U$ is in extended Foata normal form with $[U]=x\cdot y$,
  Lemma~\ref{L-fnf-product} implies the existence of sets
  $B_i\subseteq A_i$ for  $i\in[n]$ such that
  \begin{enumerate}[(a)]
  \item the word
    $X:=(A_1\setminus B_1)\,(A_2\setminus B_2)\,\cdots\,(A_n\setminus
    B_n)$ is in extended Foata normal form,
  \item $B_i\parallel A_j\setminus B_j$ for all $1\le i<j\le n$, and
  \item
    $x=\bigl[(A_1\setminus B_1)\,(A_2\setminus B_2)\,\cdots\,(A_n\setminus
    B_n)\bigr]$ and $y=[B_1\,B_2\,\cdots\,B_n]$.
  \end{enumerate}

  The words $X$ and $V$ both are in extended Foata normal form with
  $([X],[V])=(x,z)\in\cR$. Our assumption on the NFA $\cA_1$ implies
  that the word from (ii) belongs to $L(\cA_1)$, i.e., there are
  $(\iota_1,q_1)\in I_1\times F_1$ such that (ii) holds.

  Let $w_i\in[B_i]$ for $i\in[n]$. Then $[w_1w_2\cdots
  w_n]=y\in\cK$. Hence the word $w_1\cdots w_n$ is accepted by the NFA
  $\cA_2$, i.e., there are $(\iota_2,q_2)\in I_2\times F_2$ such that
  (iii) holds.

  Setting $Y_1=A_n\setminus B_n$ and  $Y_2=\bigcup_{i\in[n]}B_i$ 
  ensures that also (i) and (iv) hold. Thus, we have
  \eqref{eq:proof-of-} and therefore $W\in L(\cA)$. But this
  finishes the verification of the implication ``$\Leftarrow$'' in
  \eqref{eq:claim}.
\end{proof}

Lemma~\ref{L-fnf-vs-rational}, relates fnf-automatic and rational
relations. The following corollary describes the relation between
fnf-automatic and recognizable languages.

\begin{corollary}\label{C-recognizable->fnf-automatic}
  Any recognizable language in $\bM$ is fnf-automatic, but there are
  fnf-automatic languages that are not recognizable.
\end{corollary}

\begin{proof}
  Let $\cL\subseteq\bM$ be recognizable.  The relation $\cR=\{(1,1)\}$
  is fnf-automatic. Hence, by Lemma~\ref{L-fnf-sr}, the relation
  $\cR\cdot\bigl(\cL\times\{1\}\bigr)=\{(u,1)\mid u\in\cL\}$ is
  fnf-automatic. Since the class of fnf-automatic relations is closed
  under projections, the language $\cL$ is fnf-automatic.
  
  For the second claim, let $a,b\in\Sigma$ with $a\parallel b$ and
  $A=\{a,b\}$. Then $\cL=\{[A^n]\mid n\in\bN\}$ is fnf-automatic, but
  it is not recognizable since $\cL=\{[a^nb^n]\mid n\in\bN\}$.
\end{proof}

\section{Reachability and fnf-automaticity}

The aim of this section is the proof of the following result.

\begin{theorem}\label{T-reachability-fnf-automatic}
  Let $\cP=(Q,\Sigma,D,\Delta)$ be a loop-connected tPDS and
  $p,q\in Q$ two states. Then the one-step relation $\step_{p,q}$ and the
  reachability relation $\reach_{p,q}$ are effectively fnf-automatic.
\end{theorem}

Note that $(s,t)\in\step_{p,q}$ iff there is a trace $x$ and a
transition $(p,a,w,q)\in\Delta$ with $s=x\cdot[a]$ and $t=x\cdot[w]$. Hence
\[
  \step_{p,q}=\bigcup_{(p,a,w,q)\in\Delta}\Id_\bM\cdot
  \biggl(\bigl\{[a]\bigr\}\times\bigl\{[w]\bigr\}\biggr)\,.
\]
Since this is a finite union and since unary languages are
recognizable, the fnf-automaticity of $\step_{p,q}$ follows easily
from Lemma~\ref{L-fnf-sr}. 

The following proof of the claim regarding $\reach_{p,q}$ is more
involved.  It uses constructions and results from \cite{KoeK26} that
were shown for all tPDS and it then refines them for the case of
loop-connected systems. It starts by first demonstrating the theorem
for two special cases.

The first special case assumes that all transitions of $\cP$ shorten
the pushdown, i.e., write the empty word. Since any such system is
trivially loop-connected, we do not require this explicitly.

\begin{proposition}\label{prop:read}
  Let $\cP=(Q,\Sigma,D,\Delta)$ be a tPDS with
  $\Delta\subseteq Q\times \Sigma\times\{\varepsilon\}\times Q$, and
  $p,q\in Q$ two states. Then the relation $\reach_{p,q}$ is
  effectively fnf-automatic.
\end{proposition}

\begin{proof}
  The proof of \cite[Prop.~5.4]{KoeK26} shows
  $\reach_{p,q}=\Id_\bM\cdot(\cK\times\{1\})$ for some effectively
  recognizable language $\cK\subseteq\bM$. Hence the claim follows
  from Lemma~\ref{L-fnf-sr}.
\end{proof}

While above we consider systems that never write anything onto the
pushdown, we next want to study systems that always write
something. In other words, these systems do not shorten the pushdown
in any step. We make the additional assumption that all letters read
from the pushdown are twins, i.e., if
$(p,a,w,q),(p',a',w',q')\in\Delta$ are transitions, then
$D(a)=D(a')$. In \cite[Proof of Prop.~4.2]{KoeK26}, it is shown that
the reachability relation $\reach_{p,q}$ in such systems equals
\[
  \bigcup_{a\in\Sigma}\Id_\bM\cdot(\{[a]\}\times\cH_a)
  \underbrace{{}\cup{\Id_\bM\,}}_{\text{if }p=q}
\]
for some effectively rational sets of traces $\cH_a$. The additional
assumption of loop-connectedness allows to show that these languages
$\cH_a$ are even recognizable.

\begin{lemma}\label{lem:write1}
  Let $\cP=(Q,\Sigma,D,\Delta)$ be a loop-connected tPDS with
  $\Delta\subseteq Q\times \twins(a)\times \Sigma^+\times Q$ for some
  $a\in \Sigma$, and $p,q\in Q$ be two states. There exists an
  effectively recognizable trace language $\cH_a\subseteq \bM$ such
  that, for any $x\in\bM$,
  \begin{equation}
    \label{eq:write1}
    \{y\in\bM\mid(x\cdot[a],y)\in\reach_{p,q}\}=\{x\}\cdot\cH_a\,.
  \end{equation}
\end{lemma}

\begin{proof}
  Suppose $\cP$ is a tPDS that is not necessarily
  loop-connected. Then \cite[Lemma~5.5]{KoeK26} ensures the effective
  existence of a \emph{rational} set $\cH_a$ with
  \eqref{eq:write1}. Here, we show that this set is even
  \emph{recognizable} provided $\cP$ is loop-connected.

  So recall the construction of an $\varepsilon$-NFA
  $\cA=(Q_\cA,A,I,\delta,F)$ from the proof of
  \cite[Lemma~5.5]{KoeK26} (the set $\cH_a$ equals $[L(\cA)]$).

  We start with the only accepting state $(q,\varepsilon)$ and all
  pairs $(r,c)\in Q\times \Sigma$ as further states (recall that $Q$
  is the set of states of the tPDS $\cP$). To start, we add
  $c$-labeled transitions from $(q,c)$ to $(q,\varepsilon)$ for any
  letter $c\in \Sigma$. Then, for any transition
  $(r,c,udv,s)\in\Delta$ with $d\in \Sigma$ and $u,v\in \Sigma^*$ such
  that $d\parallel v$, we add (introducing auxiliary states) a
  $uv$-labeled path from $(r,c)$ to $(s,d)$. The set of initial states
  is $I=\{(p,a)\}$ and the set of final states is
  $F=\{(q,\varepsilon)\}$.

  We set $H_a=L(\cA)$ and $\cH_a=[H_a]$. In the proof of
  \cite[Lemma~5.5]{KoeK26}, it is then shown that
  Eq.~\eqref{eq:write1} holds for all $x\in\bM$. It therefore remains
  to be shown that $\cH_a$ is recognizable (and only this proof
  requires the tPDS $\cP$ to be loop-connected).

  To this aim, we first verify that the $\varepsilon$-NFA $\cA$
  satisfies the following: if $q$ is a state of $\cA$ and
  $w\in\Sigma^*$ such that $q\xrightarrow{w}_\cA q$, then
  $\alphabet(w)$ is connected.

  So suppose $q\xrightarrow{w}_\cA q$. If $w=\varepsilon$, then
  $\alphabet(w)=\emptyset$ is connected. From now on, we 
  assume $w\neq\varepsilon$. The construction of $\cA$ implies that
  the loop $q\xrightarrow{w}_\cA q$ does not consist of auxiliary
  states, only. Hence we can assume $q=(p_0,c_0)\in
  Q\times\Sigma$. Then there are $n\ge1$ and, for all $i\in[n]$,
  $p_i\in Q$, $c_i\in \Sigma$, and $v_i\in\Sigma^*$ with
  \[
    (p_0,c_0)\xrightarrow{v_1}_\cA(p_1,c_1)\xrightarrow{v_2}_\cA
    \cdots
    \xrightarrow{v_n}_\cA(p_n,c_n)=(p_0,c_0)
  \]
  such that only auxiliary states are visited between
  $(p_{i-1},c_{i-1})$ and $(p_i,c_i)$ for any $i\in[n]$. Furthermore,
  $w=v_1\,v_2\,\cdots\,v_n$.

  For all $i\in[n]$, the path
  $(p_{i-1},c_{i-1})\xrightarrow{v_i}_\cA(p_i,c_i)$ visits auxiliary
  states, only. Hence there are $x_i,y_i\in\Sigma^*$ such that
  \[
    (p_{i-1},c_{i-1},x_ic_iy_i,p_i)\in\Delta\,,
    c_i\parallel y_i\,,\text{ and }
    v_i=x_iy_i\,,
  \]
  i.e., these transitions form a loop as in the definition of
  loop-connectedness (cf.~Def.~\ref{Def-loop-connected}). Following
  that definition, we showed $y_i=\varepsilon$ for all $i\in[n]$ and
  therefore
  \[
    \alphabet(x_1x_2\cdots x_n)=\alphabet(v_1v_2\cdots v_n)=\alphabet(w)\,.
  \]
  Hence the loop-connectedness implies that $\alphabet(w)$ is connected.
 
  Thus, indeed, for any state $q$ of the $\varepsilon$-NFA $\cA$ and
  any word $w\in\Sigma^*$ with $q\xrightarrow{w}_\cA q$, the set
  $\alphabet(w)$ is connected. For automata with this property,
  \cite[Lemma~2.1]{Kus07} as well as \cite[Prop~5]{MusP99} explicitly
  say that the set of traces
  \[
    \cH_a=\{[w]\mid w\in L(\cA)\}=\{[w]\mid w\in H_a\}
  \]
  is effectively recognizable.
\end{proof}

\begin{proposition}\label{prop:write}
  Let $\cP=(Q,\Sigma,D,\Delta)$ be a loop-connected tPDS with
  $\Delta\subseteq Q\times \twins(b)\times \Sigma^+\times Q$ for some
  $b\in \Sigma$, and $p,q\in Q$ two states. Then the relation
  $\reach_{p,q}$ is effectively fnf-automatic.
\end{proposition}

\begin{proof}
  From  Lemma~\ref{lem:write1}, we get
  \[
    \reach_{p,q}=\bigcup_{a\in\twins(b)}\Id_\bM\cdot(\{[a]\}\times\cH_a)
    \underbrace{\cup\, \Id_\bM}_{\text{if }p=q}
  \]
  for some effectively recognizable languages
  $\cH_a\subseteq\bM$. Recall that $\Id_\bM$ is fnf-automatic and that
  the class of fnf-automatic relations is effectively closed under
  union. Hence Lemma~\ref{L-fnf-sr} implies that $\reach_{p,q}$
  is effectively fnf-automatic.
\end{proof}

We now come to the proof of the central result of this section.\bigskip

\noindent\textbf{Proof of Theorem~\ref{T-reachability-fnf-automatic}:}\\
  Above, we already showed that $\step_{p,q}$ is fnf-automatic. So it
  remains to consider the relation $\reach_{p,q}$. We first split the
  set of transitions of the loop-connected system $\cP=(Q,\Delta)$
  into finitely many sets:
  \begin{align*}
  \Delta_\varepsilon &= \Delta\cap(Q\times \Sigma\times\{\varepsilon\}\times Q)
    & \cP_\varepsilon &= (Q,\Delta_\varepsilon)\\
  \Delta_T &= \Delta \cap (Q\times T\times \Sigma^+\times Q)
    & \cP_T &= (Q,\Delta_T)
  \end{align*}
  where $T=\twins(a)$ for some $a\in\Sigma$, i.e.,
  $T=\{b\in\Sigma\mid D(a)=D(b)\}$. For $r,s\in Q$, let
  $\reach_{r,s}^\varepsilon$ denote the reachability relation of the
  system $\cP_\varepsilon$ and $\reach_{r,s}^T$ that of the system
  $\cP_T$.

  It is easily seen that $\cP_\varepsilon$ and $\cP_T$ are tPDS.  Note
  that all loops in $\cP_\varepsilon$ or $\cP_T$ are also loops in the
  system $\cP$. Hence $\cP_\varepsilon$ and $\cP_T$ are loop-connected tPDS
  such that Propositions~\ref{prop:read} and \ref{prop:write},
  respectively, are applicable. Hence the relations
  $\reach_{r,s}^\varepsilon$ and $\reach_{r,s}^T$ are effectively
  fnf-automatic.

  In the proof of \cite[Theorem~5.12]{KoeK26}, it is shown that the
  reachability relation $\reach_{p,q}$ of the tPDS $\cP$ can
  effectively be obtained from the relations
  $\reach_{r,s}^\varepsilon$ and $\reach_{r,s}^T$ by composition and
  union. Since the class of fnf-automatic structures is effectively
  closed under these operations, the theorem's claim follows.
\endproof

\section{Proof of Theorem~\ref{T-semi-main}}
\label{S-proof}

We have to demonstrate that the first-order theory of the structure
$\cS(\cP)$ is uniformly decidable for all loop-connected tPDS. To this
aim, we prove that the structure $\cS(\cP)$ is effectively automatic
\cite{Hod82,KhoN95,BluG00} for any loop-connected tPDS
$\cP=(Q,\Sigma,D,\Delta)$.

So let $\eFNF$ denote the set of all words in extended Foata normal
form. Then $\eFNF$ is regular. By
Theorem~\ref{T-reachability-fnf-automatic}, also the languages
$L_{\step_{p,q}}$ and $L_{\reach_{p,q}}$ for $p,q\in Q$ are
effectively regular. Recall that the mapping $\eFNF\to\bM$ that maps a
word $W\in\eFNF$ to the trace $[W]$ is surjective. Hence the structure
$\cS(\cP)$ is effectively automatic.

  Now the claim of the theorem follows since the first-order theory of
  automatic structures is uniformly decidable~\cite[Cor.~4.2]{KhoN95}
  (alternatively, see \cite{Hod82} or \cite[Thm.~2.1]{BluG04}). 
  \endproof
  
\begin{remark}
  Let $\cR\subseteq\bM^k$ be a $k$-ary fnf-automatic relation and
  $(p_1,\dots,p_k)$ a tuple of states of the loop-connected tPDS
  $\cP$. Extend the configuration graph $\ConfGraph(\cP)$ with the relation
  \[
    \{(p_i,t_i)_{1\le i\le k}\mid (t_1,\dots,t_k)\in\cR\}
  \]
  and the structure $\cS(\cP)$ with the relation $\cR$. Then the above
  proof still gives the decidability of the first-order theory. Adding
  such relations seems pretty exotic, but it might be of interest in
  the following two special cases:
  \begin{itemize}
  \item Let $\cL\subseteq\bM$ be an fnf-automatic (e.g., recognizable)
    set of traces understood as a property of the content of the
    pushdown.  Adding $Q\times\cL$ as a unary relation to
    $\ConfGraph(\cP)$ allows to use this property of pushdown contents
    in formulas. Then formulas can make statements of the form
    ``There is a path from here to there that visits some state from
    $Q\times\cL$.''
  \item Let $\Delta'\subseteq\Delta$ be some set of transitions such
    that also the restricted tPDS $\cP'=(Q,\Sigma,D,\Delta')$ is
    loop-connected. Since the reachability relation of this restricted
    tPDS is fnf-automatic, one can add it as a binary relation to the
    configuration graph. Then formulas can express properties of the
    form ``There is a path from here to there that does not use
    forbidden transitions from $\Delta\setminus\Delta'$.''
  \end{itemize}
  Recall that in case $D=\Sigma^2$, even the monadic second-order
  theory of $\ConfGraph(\cP)$ is decidable. By Example~\ref{E-MSO},
  this does not hold for arbitrary loop-connected tPDS. But another
  restriction of second-order logic yields a decidable theory of
  reachability. In this restriction, second-order quantification
  $\exists X\text{ infinite}\colon\varphi$ is possible for relational
  variables $X$ of arbitrary arity $r$, but with the condition that
  the atomic formula $(x_1,\dots,x_r)\in X$ appears only negatively in
  $\varphi$. In \cite{KusL10}, it is shown that the validity of such
  formulas is decidable for any automatic structure. Thus, this
  decidability also holds for the configuration graph with
  reachability of a loop-connected tPDS.

  The work on automatic structures revealed that additional types of
  quantifiers are possible: the infinity quantifier $\exists^\infty$
  \cite{BluG00,BluG04}, modulo-counting quantifiers \cite{Rub04},
  Ramsey-quantifier \cite{Rub08}, and a boundedness quantifier
  \cite{Kus09a}.
\end{remark}

\bibliographystyle{fundam}
\bibliography{references}

@Book{Ber79,
  author = {J. Berstel},
  title = {Transductions and context-free languages},
  year = {1979},
  publisher = {Teubner Studienb{\"u}cher},
  address = {Stuttgart}
}

@InProceedings{BluG00,
  author = {A. Blumensath and E. Gr{\"a}del},
  title = {Automatic {S}tructures},
  booktitle = {LICS'00},
  year = {2000},
  pages = {51--62},
  publisher = {IEEE Computer Society Press}
}

@Article{BluG04,
  author = {A. Blumensath and E. Gr{\"a}del},
  title = {Finite presentations of infinite structures: Automata and
  interpretations},
  journal = {Theory of Computing Systems},
  year = {2004},
  volume = {37},
  number = {6},
  pages = {641-674}
}

@Book{CarF69,
  author = {Cartier, P. and D. Foata},
  title = {Probl\`emes combinatoires de commutation et r\'earrangements},
  year = {1969},
  series = {Lecture Notes in Mathematics vol.\ 85},
  publisher = {Springer, Berlin - Heidelberg - New York}
}

@Book{CouE12,
  author = {B. Courcelle and Joost Engelfriet},
  title = {Graph structure and monadic second-order logic},
  year = {2012},
  publisher = {Cambridge University Press}
}

@Book{DieR95,
  author = {Diekert, V. and G. Rozenberg},
  title = {The Book of Traces},
  year = {1995},
  publisher = {World Scientific Publ.\ Co.}
}

@Article{EldKO08,
  author = {M. Elder and M. Kambites and G. Ostheimer},
  title = {On groups and counter automata},
  journal = {International Journal of Algebra and Computation},
  year = {2008},
  volume = {18},
  number = {08},
  pages = {1345-1364}
}

@Article{ElsO04,
  author = {G.Z. Elston and G. Ostheimer},
  title = {On groups whose word problem is solved by a counter automaton},
  journal = {Theoretical Computer Science},
  year = {2004},
  volume = {320},
  number = {2-3},
  pages = {175-185}
}

@InProceedings{Gil94,
  author = {R.H. Gilman},
  title = {Formal languages and infinite groups},
  booktitle = {Geometric and Computational Perspectives on Infinite Groups},
  year = {1994},
  series = {{DIMACS} Series in Discrete Mathematics and Theoretical Computer
  Science},
  volume = {25},
  pages = {27--51},
  publisher = {{DIMACS/AMS}}
}

@Article{Hod82,
  author = {B.R. Hodgson},
  title = {On direct products of automaton decidable theories},
  journal = {Theoretical Computer Science},
  year = {1982},
  volume = {19},
  pages = {331-335}
}

@Book{Hod93,
  author = {W. Hodges},
  title = {Model Theory},
  year = {1993},
  publisher = {Cambridge University Press}
}

@Article{IbaSK76,
  author = {O.H. Ibarra and S.K. Sahni and C.E. Kim.},
  title = {Finite automata with multiplication},
  journal = {Theoretical Computer Science},
  year = {1976},
  volume = {2},
  number = {3},
  pages = {271-294}
}

@Article{ItoMVM01,
  author = {Masami Ito and Carlos Mart{\'i}n-Vide and Victor Mitrana},
  title = {Group weighted finite transducers},
  journal = {Acta Informatica},
  year = {2001},
  volume = {38},
  pages = {117-129}
}

@Article{Kam09,
  author = {Kambites, M.},
  title = {Formal Languages and Groups as Memory},
  journal = {Communications in Algebra},
  year = {2009},
  volume = {37},
  pages = {193--208}
}

@InProceedings{KhoN95,
  author = {B. Khoussainov and A. Nerode},
  title = {Automatic presentations of structures},
  booktitle = {Logic and Computational Complexity},
  year = {1995},
  series = {Lecture Notes in Comp.\ Science vol.\ 960},
  pages = {367-392},
  publisher = {Springer}
}

@InCollection{Kle56,
  author = {S.C. Kleene},
  title = {Representation of events in nerve nets and finite automata},
  booktitle = {Automata Studies},
  editor = {C.E. Shannon and J. McCarthy},
  year = {1956},
  series = {Annals of Mathematics Studies vol.\ 34},
  pages = {3-40},
  publisher = {Princeton University Press}
}

@InProceedings{KoeK23,
  author = {Ch. K{\"o}cher and D. Kuske},
  title = {Forwards- and Backwards-Reachability for Cooperating Multi-Pushdown
  Systems},
  booktitle = {FCT'23},
  year = {2023},
  series = {Lecture Notes in Comp.\ Science vol.~14252},
  pages = {318-332},
  publisher = {Springer}
}

@Article{KoeK25a,
  author = {Ch. K{\"o}cher and D. Kuske},
  title = {Backwards-Reachability for Cooperating Multi-Pushdown Systems},
  journal = {Journal of Computer and System Sciences},
  year = {2025},
  volume = {148},
  note = {article no. 103601}
}

@Article{KoeK26,
  author = {Ch. K{\"o}cher and D. Kuske},
  title = {Reachability in trace-pushdown systems},
  journal = {Theoretical Computer Science},
  year = {2026},
  volume = {1077},
  pages = {115971}
}

@Article{Kus07,
  author = {D. Kuske},
  title = {Weighted asynchronous cellular automata},
  journal = {Theoretical Computer Science},
  year = {2007},
  volume = {374},
  pages = {127-148}
}

@InProceedings{Kus09a,
  author = {D. Kuske},
  title = {Theories of automatic structures and their complexity},
  booktitle = {CAI 2009},
  year = {2009},
  series = {Lecture Notes in Comp.\ Science vol.\ 5725},
  pages = {81-98},
  publisher = {Springer}
}

@InProceedings{Kus23,
  author = {D. Kuske},
  title = {A class of rational trace relations closed under composition},
  booktitle = {{FSTTCS}'23},
  year = {2023},
  series = {Leibniz International Proceedings in Informatics (LIPIcs) vol.~
  284},
  pages = {20:1-20:20},
  publisher = {Schloss Dagstuhl - Leibniz-Zentrum f{\"{u}}r Informatik}
}

@InProceedings{Kus25,
  author = {D. Kuske},
  title = {The theory of reachability in trace-pushdown systems},
  booktitle = {CiE'25},
  year = {2025},
  series = {Lecture Notes in Comp.\ Science.\ vol.\ 15764},
  pages = {283-298},
  publisher = {Springer}
}

@Article{KusL10,
  author = {D. Kuske and M. Lohrey},
  title = {Some natural decision problems in automatic graphs},
  journal = {Journal of Symbolic Logic},
  year = {2010},
  volume = {75},
  number = {2},
  pages = {678-710}
}

@Article{MitS01,
  author = {V. Mitrana and R. Stiebe},
  title = {Extended finite automata over groups},
  journal = {Discrete Applied Mathematics},
  year = {2001},
  volume = {108},
  number = {3},
  pages = {287-300}
}

@InProceedings{MusP99,
  author = {A. Muscholl and D. Peled},
  title = {Message sequence graphs and decision problems on {M}azurkiewicz
  traces},
  booktitle = {MFCS'99},
  year = {1999},
  series = {Lecture Notes in Comp.\ Science vol.~ 1672},
  pages = {81-91},
  publisher = {Springer}
}

@InProceedings{OsuMZ16,
  author = {E. {D'Osualdo} and R. Meyer and G. Zetzsche},
  title = {First-order logic with reachability for infinite-state systems},
  booktitle = {LICS'16},
  year = {2016},
  pages = {457--466},
  publisher = {{ACM}}
}

@Article{Pos46,
  author = {E.L. Post},
  title = {A variant of a recursively unsolvable problem},
  journal = {Boll Amer. Math. Soc.},
  year = {1946},
  volume = {52},
  number = {4},
  pages = {264-269}
}

@Article{REnK09,
  author = {E. Render and M. Kambites},
  title = {Rational subsets of polycyclic monoids and valence automata},
  journal = {Information and Computation},
  year = {2009},
  volume = {207},
  number = {11},
  pages = {1329-1339}
}

@Article{Rab69,
  author = {Rabin, M.O.},
  title = {Decidability of second-order theories and automata on infinite
  trees},
  journal = {Trans.\ Amer.\ Math.\ Soc.},
  year = {1969},
  volume = {141},
  pages = {1-35}
}

@PhdThesis{Ren10,
  author = {Elaine Render},
  title = {Rational Monoid and Semigroup Automata},
  year = {2010},
  school = {University of Manchester}
}

@PhdThesis{Rub04,
  author = {S. Rubin},
  title = {Automatic {S}tructures},
  year = {2004},
  school = {University of Auckland}
}

@Article{Rub08,
  author = {S. Rubin},
  title = {Automata presenting structures: A survey of the finite string case},
  journal = {Bulletin of Symbolic Logic},
  year = {2008},
  volume = {14},
  pages = {169-209}
}

@InProceedings{Zet15,
  author = {G. Zetzsche},
  title = {The Emptiness Problem for Valence Automata or: Another Decidable
  Extension of {P}etri Nets},
  booktitle = {{RP}'15},
  year = {2015},
  series = {Lecture Notes in Comp.\ Science vol.~9328},
  pages = {166--178},
  publisher = {Springer}
}

@Article{Zet21,
  author = {G. Zetzsche},
  title = {The emptiness problem for valence automata over graph monoids},
  journal = {Inf. Comput.},
  year = {2021},
  volume = {277},
  note = {article no.~104583}
}


\end{document}